%% file: main.tex
\documentclass[a4paper,UKenglish]{lipics-v2019}

\nolinenumbers

\usepackage{latexsym,graphicx,mathrsfs}
\usepackage{amsthm,amsmath,amssymb,enumerate}
\usepackage{empheq}
\usepackage{xspace}
\usepackage{bm}
\usepackage{ifpdf}
\usepackage{algorithm}
\usepackage[noend]{algpseudocode}
\usepackage{nicefrac}
\usepackage{thm-restate}
\usepackage{mathtools}

\allowdisplaybreaks

\definecolor{Darkblue}{rgb}{0,0,0.4}
\definecolor{Brown}{cmyk}{0,0.81,1.,0.60}
\definecolor{Purple}{cmyk}{0.45,0.86,0,0}

\ifpdf
 
\fi

\newcommand{\junk}[1]{}
\newcommand{\ignore}[1]{}

\newcommand{\R}[0]{{\ensuremath{\mathbb{R}}}}

\def\ceil#1{\lceil #1 \rceil}

\newcommand{\poly}{\operatorname{poly}}

\newcommand{\sse}{\subseteq}

\newcommand{\F}{{\mathbb{F}}}
\newcommand{\barF}{\overline{F}}

\renewcommand{\emptyset}{\varnothing}

\newcommand{\e}{\varepsilon}
\newcommand{\eps}{\varepsilon}

\newcommand{\agnote}[1]{}
\newcommand{\aknote}[1]{}
\newcommand{\jlnote}[1]{}
\newcommand{\elnote}[1]{}

\newcommand{\qedsymb}{\hfill{\rule{2mm}{2mm}}}

\newcommand{\initOneLiners}{%
    \setlength{\itemsep}{0pt}
    \setlength{\parsep }{0pt}
    \setlength{\topsep }{0pt}
}

\newcommand{\squishlist}{
 \begin{list}{$\bullet$}
  { \setlength{\itemsep}{0pt}
     \setlength{\parsep}{3pt}
     \setlength{\topsep}{3pt}
     \setlength{\partopsep}{0pt}
     \setlength{\leftmargin}{1.5em}
     \setlength{\labelwidth}{1em}
     \setlength{\labelsep}{0.5em} } }

\newcommand{\squishend}{
  \end{list}  }

\newcommand{\lcc}{\lceil\!\!\lceil}
\newcommand{\rcc}{\rceil\!\!\rceil}

\DeclarePairedDelimiterX{\infdivx}[2]{(}{)}{%
  #1\;\delimsize\|\;#2%
}

\newcommand{\open}{\small{\mathsf{open}}}
\newcommand{\cost}{\small{\mathsf{cost}}}
\newcommand{\impr}{\small{\mathsf{improv}}}

\newcommand{\flnumber}{{\alpha_{\mathsf{FL}}}\xspace}

\newcommand{\kmed}{\textsc{$k$-Median}\xspace}
\newcommand{\kmeans}{\textsc{$k$-Means}\xspace}
\newcommand{\matmed}{\textsc{Matroid Median}\xspace}
\newcommand{\fl}{\textsc{Facility Location}\xspace}
\newcommand{\mkc}{\textsc{Max $k$-Coverage}\xspace}
\newcommand{\lc}{\textsc{Label Cover}\xspace}
\newcommand{\OPT}{\mathsf{OPT}}

\newcommand{\kclique}{\textsc{$k$-Clique}\xspace}
\newcommand{\kbiclique}{\textsc{$k$-Biclique}\xspace}
\newcommand{\kdominatingset}{\textsc{$k$-Dominating Set}\xspace}

\newcommand{\E}{\mathbb{E}}

\newcommand{\cald}{\mathcal{D}}
\newcommand{\cali}{\mathcal{I}}
\newcommand{\call}{\mathcal{L}}
\newcommand{\calu}{\mathcal{U}}
\newcommand{\cals}{\mathcal{S}}

\newcommand{\f}{\frac}

\newcommand{\bn}{\binom}

\newcommand{\lds}{\ldots}

\newcommand{\s}{\subseteq}

\newcommand{\logn}{\log n}

\newcommand{\De}{\Delta}

\newcommand{\be}{\beta}

\newcommand{\el}{\ell}
\newcommand{\I}{{\cal I}}

\newcommand{\CoresetSize}{\ensuremath{O(\e^{-2} k \log n)}\xspace}

\title{{\bf Tight FPT Approximations for $k$-Median and
    $k$-Means}}

\author{Vincent Cohen-Addad}{Universit\'e Pierre et Marie Curie, Paris}{}{}{}{}
\author{Anupam Gupta}{Carnegie Mellon University}{}{}{Supported in part by NSF awards CCF-1536002, CCF-1540541, and CCF-1617790.}
\author{Amit Kumar}{IIT Delhi}{}{}{}
\author{Euiwoong Lee}{New York University}{}{}{Supported in part by the Simons Collaboration on Algorithms and Geometry.}
\author{Jason Li}{Carnegie Mellon University}{}{}{Supported in part by NSF awards CCF-1536002, CCF-1540541, and CCF-1617790.}

\authorrunning{V.\,Cohen-Addad, A.\,Gupta, A.\,Kumar, E.\,Lee, and J.\,Li}
\Copyright{Vincent Cohen-Addad, Anupam Gupta, Amit Kumar, Euiwoong Lee, and Jason Li}



\acknowledgements{
We thank Deeparnab Chakrabarty, Ola Svensson, and Pasin Manurangsi for useful discussions.
This research was partially conducted when A.\ Kumar was visiting A.\ Gupta and Carnegie Mellon
University as part of the Joint Indo-US Virtual Center for Algorithms under Uncertainty.
}

\EventEditors{John Q. Open and Joan R. Access}
\EventNoEds{2}
\EventLongTitle{42nd Conference on Very Important Topics (CVIT 2016)}
\EventShortTitle{CVIT 2016}
\EventAcronym{CVIT}
\EventYear{2016}
\EventDate{December 24--27, 2016}
\EventLocation{Little Whinging, United Kingdom}
\EventLogo{}
\SeriesVolume{42}
\ArticleNo{23}

\begin{document}

\newtheorem{conjecture}[theorem]{Conjecture}
\newtheorem{fact}[theorem]{Fact}
\newtheorem{observation}[theorem]{Observation}
\newtheorem{subclaim}[theorem]{Subclaim}
\newtheorem{construction}[theorem]{Construction}
\newtheorem{reduction}[theorem]{Reduction}
\newtheorem{invariant}{Invariant}
\newtheorem{extension}[theorem]{Extension}
\numberwithin{algorithm}{section}

\maketitle

\begin{abstract}
  We investigate the fine-grained complexity of approximating the
  classical \kmed/\kmeans clustering problems in general metric spaces. We show how to improve
  the approximation factors to $(1+2/e+\e)$ and $(1+8/e+\e)$ respectively, using algorithms that run in fixed-parameter
  time. Moreover, we show that we cannot do better in FPT time, modulo
  recent complexity-theoretic conjectures. 
\end{abstract}



\setcounter{page}{1}

\input{intro}

\input{prelims}

\section{The Approximation Algorithm}
\label{sec:approx}

We now give the $(1+2/e+ \eps)$-approximation algorithm for \kmed, where
$\eps > 0$ is a fixed parameter throughout this section. The running
time of the algorithm is $f(k, \eps) \cdot \poly(n)$, where
$f(k, \eps) = O(\e^{-2} k \log k)^k$. We then indicate the alterations
to get algorithms for \kmeans and \matmed.

\subsection{The Intuition}
\label{sec:intuition}

We focus on \kmed for now; the ideas for the other problems are
analogous. The first idea is to reduce the size of the client set $C$ to
\CoresetSize---this can be done by results on core-sets for \kmed, which
consolidate the clients into a small number of distinct
locations~\cite{Chen,FL}. The consolidated clients now have weights, but
this extension to weighted \kmed does not pose a problem.

The next idea is to carefully enumerate over the structure of an optimal
solution. Consider an optimal solution
$F^\star =\{f^\star_1, \ldots, f^\star_k\}$. For a facility
$f^\star_i \in F^\star$, let ``cluster'' $C^\star_i$ be the clients
assigned to $f^\star_i$, i.e., the subset of clients $C$ for which
$f^\star_i$ is closest open facility. Let $\ell_i$ be the client in
$C^\star_i$ closest to $f^\star_i$ -- we call it the {\em leader} of
cluster $C^\star_i$. Let $R_i$ be the distance $d(f^\star, \ell_i)$,
suitably discretized. Our algorithm guesses the leaders $\ell_i$ and the
distances $R_i$ for each $i \in [k]$. Since the size of $C$ is
\CoresetSize, there are $(O(\e^{-2}k \log n))^k$ choices for
leaders,\footnote{Our analysis will tighten this bound to
  $O(\e^{-2}\log n)^k$, but this improvement can be ignored for this
  intuition section.} and a similar number of choices for the distances;
moreover, this quantity can be shown to be $f(k, \eps) \cdot n^{O(1)}$.

Assume now that we have correctly guessed the leaders $\ell_i$ and
distances $R_i$. For each leader $\ell_i$, let $F_i$ be the facilities
at distance about $R_i$ from $\ell_i$---this set $F_i$ contains
$f^\star_i$. By making copies, assume the sets $F_i$ are disjoint. Now
our task is to select one facility from each set $F_i$ such that the
total (weighted) assignment cost of the clients in $C$ is minimized. As
such, this seems like a decreasing \emph{supermodular minimization}
problem with a (partition) matroid constraint. (Observe that choosing an
arbitrary center in each $F_i$ gives us a $3$-approximation in FPT time,
but we want to do much better.)

The last idea is to convert this into a monotone submodular maximization
problem, again with a partition matroid constraint. For each set $F_i$,
we add a \emph{fictitious facility} $f_i'$ such that (i)~the assignment cost of
clients to the fictitious facilities is at most $3 OPT$, and (ii)~for a
subset $S$ of facilities, the ``improvement''
$\cost(C, F') - \cost(C, F' \cup S)$, where $F'$ is the set of fictitious
facilities, is a monotone submodular function. We finally show that a
$(1-1/e)$-approximation for this submodular maximization problem gives
the desired approximation guarantee.  The next two sections describe the
algorithm for \kmed in detail. The extension to \kmeans, \matmed and \fl then
appears in \S\ref{sec:extensions}.

\subsection{Client Reduction via Coresets}
\label{sec:coresets}

Consider an instance $\I = ((V,d),C,\F,k)$ of the \kmed problem. Let
$\e > 0$ be a fixed constant. We now define the notion of core-sets and
use known results to reduce the size of $C$ to (a weighted) a set of
size \CoresetSize.

\begin{definition}[Core-set]
  A \emph{(strong) core-set} for $\I$
  is a set of clients $C' \sse V$ along with weights $w_j$ for all
  $j \in C'$, such that
  \[ \sum_{j \in C'} w_j \; d(j, F)
    \in (1-\e, 1+\e) \cdot \sum_{j \in C}  d(j, F), \]
  for every $F \sse \F$ with $|F| = k$.
\end{definition}
A similar definition holds for a strong core-set for the \kmeans
problem. Since we deal only with strong core-sets in this paper, we drop
the modifier and refer to them only as core-sets. The first core-sets
for metric \kmed were given by Chen~\cite{Chen}; the following result is the best current
construction:
\begin{theorem}[\cite{FL}, Theorem~15.4]
  For $0 \leq \e, \delta \leq 1/2$, there exists a Monte Carlo algorithm
  that for each instance $I$ of \kmed on a general metric, outputs a
  core-set $C' \sse C$ with size
  \[ |C'| = O\Big(\frac{k \log n + \log \nicefrac1\delta}{\e^2}\Big) \]
  with probability $1-\delta$, where $n = |V|$. 
  Moreover, the algorithm runs in time
  $O(k(n+k) + \log^2(1/\delta) \log^2 n)$.  For \kmeans, the core-set is
  of size $|C'| = O\big(\frac{k \log n + \log
      \nicefrac1\delta}{\e^4}\big)$, and 
  the runtime remains the same.
\end{theorem}


The power of core-sets lies in the following fact.
\begin{fact}\label{fact:coreset}
  Consider a \kmed/\kmeans instance $\I=((V,d),C,\F,k)$,
  and let $C'$ be a (strong) core-set with
  weights $w$. Consider the weighted instance $\I'=((V,d),C',\F,k,w)$,
  which is the instance $\I$ with its clients replaced by the weighted
  clients in the core-set. Then, for any $\be\ge1$, a $\be$-approximate solution $F\s \F$ to $\I'$ is a $\be(1+O(\e))$-approximate solution to $\I$. 
\end{fact}

Therefore,
in order to find a $(1+2/e+O(\e))$-approximation to a \kmed $\I$, it
suffices to find a $(1+2/e+O(\e))$-approximation to $\I'$, and
analogously for \kmeans. Henceforth, we restrict our attention to the
core-set instance $\I'$. In other words, we assume that our instances
have only a small number of clients, but now the clients have associated
weights. In the following sections, we show how to 
approximate such weighted \kmed/\kmeans instances in FPT
time.


\subsection{Reduction to Submodular Maximization}
\label{sec:reduction}

Given Fact~\ref{fact:coreset}, we only consider instances
$\I = ((V,d),C,\F,k,w)$ of weighted \kmed, where clients in $C$ have
weights in the range $[1,n]$ and $|C|$ is bounded by \CoresetSize. In
this section we prove the following approximation guarantee for \kmed;
this, combined with Fact~\ref{fact:coreset}, proves the \kmed statement
in Theorem~\ref{thm:main-alg}.

\begin{theorem}\label{thm:main-O}
Let $\e$ be a fixed parameter. Given a \kmed instance $\I=((V,d),C',\F,k,w)$ with $|C'|=O(\e^{-2}k\log n)$, there is a $(1+2/e+O(\e))$-approximation algorithm that runs in $f(k,\e)n^{O(1)}$ time.
\end{theorem}

By scaling, assume the minimum distance between points in $V$ is 1, so
the aspect ratio $\Delta$ is the maximum distance between two points in
$V$.  For a positive integer $a$, define
$\lcc a \rcc := (1 + \e)^{\ceil{\log_{(1+\e)} a}}$ as the smallest power
of $(1+\e)$ larger than or equal to $a$.  Here, $\e$ is the same fixed
parameter as the one used in the core-set.

The formal algorithm follows the intuition in \S\ref{sec:intuition} and
is described in Algorithm~\ref{alg:main}; let us step through it now. We
iterate over all possible values $\ell_1, \ldots, \ell_k$ for the
leaders, and $R_1, \ldots, R_k$ for the corresponding distances. The
same vertex could appear several times in the subset
$\{\ell_1, \ldots, \ell_k\}$, and so the latter should be thought of as
a multi-set. 
In Step~\ref{line:fict}, we add $k$ new fictitious facilities: for each
$i$, the new facility $f_i'$ is at distance $2 R_i$ from all the
facilities in $F_i$. The distance to all other points is determined by
triangle inequality in Step~\ref{line:metric}. Claim~\ref{clm:metric}
shows that this forms a valid metric. In Step~\ref{line:impr}, we define
the ``improvement'' function $\impr(S)$ as the reduction in cost due to
adding in the facilities in $S$. Claim~\ref{clm:submod} shows this
function is monotone submodular. This means we can use the
$(1-1/e)$-approximation algorithm~\cite{CCPV11} for monotone submodular
function maximization subject to a matroid constraint to find a set $S$
which contains exactly one facility from each of the sets $F_i$, since
this is a partition matroid constraint. Observe that the function
$\impr(\cdot)$ can be computed efficiently. This completes the
description of the algorithm.

\begin{algorithm}
  \caption{\tt FindCenters}
  \label{alg:main}
  \begin{algorithmic}[1]
    \For{every multi-set $\{\ell_1, \ell_2, \ldots, \ell_k\} \s C$}
    \For{every $R_1, \ldots, R_k$ such that $R_i \in [1, \ldots,
      \lcc \De \rcc]$ and $R_i$ is a power of $(1+\e)$}
    \State $F_i \gets \{ f \in F \mid \lcc d(f, \ell_i)\rcc = R_i \}$
    \State make copies of facilities to ensure that $F_1, \ldots, F_k$
    are disjoint \label{line:disjoint}
    \State Initialize $F'\gets\emptyset$; $F'$ will be the set of fictitious facilities
    \For{$i = 1, \ldots, k$}
    \State add a new fictitious facility $f_i'$ to $F'$ and set
    $d(f_i', f) := 2R_i$ for all $f \in F_i$ \label{line:fict}
    \State define 
    $d(f_i',v) := \min_{f \in F_i} (d(f_i',f) + d(f,v))$ for all other
    points $v$ \label{line:metric}
    \EndFor
    \State define  
    $\impr(S) := \cost(C, F') - \cost(C, S \cup F')$ for every set $S  \sse \F$ \label{line:impr}
    \State find $S \sse \F$ approximately maximizing $\impr(S)$, such
    that $|S \cap F_i| = 1$ \label{line:submodmax}
    \EndFor
    \EndFor
    \State among all
    sets $S$ computed in line~\ref{line:submodmax}, output $S$ for which
    $\cost(C, S)$ is minimized. \label{line:final}
  \end{algorithmic}
\end{algorithm}

To prove correctness of the algorithm, we need to show two things: the
distance function defined on $F' \cup V$ in Step~\ref{line:metric} is a
metric, and the function $\impr$ defined in Step~\ref{line:impr} is
monotone and submodular. We defer the simple proofs to
\S\ref{sec:omitted-proofs}.

\begin{restatable}[Metricity]{claim}{Metric}
  \label{clm:metric}
  Consider the set $F'$ defined during an iteration of the
  algorithm. The distance function defined on $F' \cup V$ is a metric.
\end{restatable}

\begin{restatable}[Submodularity]{claim}{Submod}
  \label{clm:submod}
  The function $\impr(S)$ defined in Step~\ref{line:impr} is monotone and submodular with $\impr(\emptyset)=0$.
\end{restatable}


Now to bound the runtime.
Since $|C|=\CoresetSize$, there are at most
$\bn{O(\e^{-2}k\log n)+k-1}k=(O(\e^{-2}\log n))^k$ different multi-sets of size $k$
with elements in $C$. In addition, there are $\log_{1+\e}\De$ many
choices for $R_i$ for each $i\in[k]$. Therefore, the number of iterations
in Step~1 of the algorithm can be bounded by
\begin{gather}
  (O(\e^{-2}\log n))^k \cdot (\log_{1+\e}\De)^k \leq \bigg(O\Big(
      \frac{(\log \De)(\logn)}{\e^{2}} \Big)\bigg)^k.
\end{gather}
As argued in \S\ref{sec:aspect-ratio}, since we started with the
unweighted \kmed problem, the aspect ratio $\Delta$ can be assumed to
polynomially bounded in $n$, and so the number of iterations can be
bounded by $(O(\log n/\e^2))^k$, which is at most
$n\cdot (O(\e^{-2} k\log k))^k$. Indeed, in case $k<\f\logn{\log\logn}$,
$(O(\log n/\e^2))^k \leq (O(1/\e^2))^k\cdot \smash{(\log
n)^{\f\logn{\log\logn}}} = (O(1/\e^2))^k\cdot n$. Else $\log n \leq O(k
\log k)$, and hence $(O(\log n/\e^2))^k = (O(k\log k/\e^2))^k$.

The algorithm for submodular maximization subject to a matroid
constraint takes polynomial time, given a value oracle for the
function~\cite[Theorem~1.1]{CCPV11}: in fact it can be sped up for the
case of partition matroid constraints~\cite[\S3.3]{CCPV11}. The value
oracle for $\impr(S)$ can itself be implemented in polynomial time.
Hence each iteration of the algorithm can be run in time polynomial in
$n$.

The submodular maximization algorithm is a randomized Monte-Carlo
algorithm that succeeds with only probability $1 - 1/n^2$, but we can
easily boost the success probability by repetition: by running it
$\tau := \poly(\e^{-1} k \log n)$ times for each input $S$ and returning
the maximum value obtained, we can ensure that  with high probability it succeeds in all the
calls we make.

\subsubsection{Approximation Ratio}

We now argue about the approximation ratio of the algorithm. We fix an
optimal solution to the instance. Let $F^\star=\{f^\star_1,\lds,f^\star_k\}$ be the
centers opened by this solution. Define $C^\star_i$ as the clients
for which the closest open center is $f^\star_i$, i.e.,
$C^\star_i := \{j\in C:d(j,f_i^\star)=d(j,F^\star)\}$. We define the notion of
leaders with respect to this solution.

\begin{definition}[Leader]
For each $i\in[k]$, call a client $j\in C^\star_i$ that minimizes $d(j,f^\star_i)$ over all $j\in C^\star_i$ the \emph{leader} $\el^\star_i$ of center $f^\star_i$. If there are multiple clients $j\in C^\star_i$ achieving the minimum, declare an arbitrary one to be the leader. Note that a client can be the leader of multiple centers $f^\star_i$. The \emph{leaders} w.r.t. the solution $F^\star$ is the multi-set $\{\el^\star_1,\lds,\el^\star_k\}$.
%
For each leader $\el^\star_i$, the \emph{radius} $R^\star_i$ is defined as $\lcc d(\el^\star_i, f^\star_i) \rcc. $
%
%
\end{definition}

Consider the iteration of Algorithm~\ref{alg:main} where
$\ell_1, \ldots, \ell_k$ are equal to
$\ell^\star_1, \ldots, \ell_k^\star$ respectively, and
$R_1, \ldots, R_k$ are equal to $R^\star_1, \ldots, R^\star_k$
respectively. Let $S^\star$ be the set output in
Step~\ref{line:submodmax} of the algorithm. It suffices to show that
$\cost(C, S^\star) \leq (1+ 2/e + \e) \cost(C, F^\star)$. We proceed
to show this in the rest of the section.

As in the algorithm, define 
\[ F_i := \{ f \in F \mid \lcc d(f, \ell^\star_i)\rcc = R^\star_i \}, \]
so that $f^\star_i\in F_i$ for each $i\in[k]$.
(Recall  that the sets $F_i$ are disjoint  by
duplicating facilities.) 
Let $F'=\{f_1', \ldots, f_k'\}$ be the set of fictitious facilities defined in the algorithm. 
%

We are interested in the solutions $S$ that consist of one center from each $F_i$, since one such solution is the desired $F^\star$. More formally, define a solution $S$ to be \emph{valid} if the set $S$ can be listed as $(f_1,\lds,f_k)$ so that $f_i\in F_i$ for each $i\in[k]$. 

\begin{claim}\label{lem:bound-S}
For every valid $S$, $\cost(C, F'\cup S) = \cost(C, S)$.
\end{claim}
\begin{proof}
List the set $S$ as $(f_1,\lds,f_k)$, where $f_i\in F_i$ for each $i\in[k]$. 
Informally, this claim amounts to showing that the fictitious facilities $F'$ do not improve the solution $S$. 
To formalize this idea, fix a client $j\in C$ and a fictitious facility $f'_i$, and let $f\in F_i$ 
be a closest center to $j$ in $F_i$.  Below, we show that in fact, client $j$ is closer to $f_i\in S$ than to $f_i'\in F'$:
\[ d(j,f_i) \stackrel{(\text{$\triangle$ ineq.})}{\le} d(j,f) + d(f,\el^\star_i)+d(\el^\star_i,f_i) \le d(j,f) + R^\star_i + R^\star_i = d(j,f)+d(f,f'_i)=d(j,f'_i).\]
Therefore, we have $d(j,F')\ge d(j,S)$ for all clients $j$, so
\[ \cost(C, F'\cup S) =\sum_{j \in C'} w_j \; d(j, F'\cup S) = \sum_{j \in C'} w_j \; d(j,S) = \cost(C,S),\]
as desired.
\end{proof}

We now bound the cost of the solution which opens facilities at $F'$. 
\begin{claim}\label{lem:bound-OPT}
$\cost(C, F') \le (3+2\e)\,\cost(C, F^\star)$.
\end{claim}
\begin{proof}
It suffices to show that $d(j,F') \le (3+2\e)\, d(j, F^\star)$ for each client $j\in C$. Fix a client $j\in C$, and let $f^\star_i\in F^\star$ be a center achieving $d(j,f^\star_i)=d(j,F^\star)$. Since $\el^\star_i$ is the leader of center $f^\star_i$, we have
\begin{gather}
d(j,f^\star_i) \ge d(\el_i,f^\star_i) \ge \f{R^\star_i}{1+\e} .\label{eq:claim2.5-lb}
\end{gather}
Recall that $f^\star_i\in F_i$. Therefore,
\begin{align*}
 d(j,F') &\le d(j,f'_i) \stackrel{(\triangle)}{\le}
           d(j,f^\star_i)+d(f^\star_i,f'_i) = d(j,f^\star_i) +
           2R^\star_i \\ &\stackrel{(\ref{eq:claim2.5-lb})}{\le} d(j,f^\star_i)+2(1+\e)\,d(j,f^\star_i)\le (3+2\e)\,d(j,F^\star),
\end{align*}
as desired.
\end{proof}

Let $S^\star$ be the set output in Step~\ref{line:final}. Since the algorithm of~\cite{CCPV11} is  $(1-1/e)$-approximation,  
\begin{gather}
  \impr(S^\star) \geq (1-1/e) \impr(F^\star) \label{eq:11}
\end{gather}

\begin{lemma}
  \label{lem:main-lemma}
The solution $S^\star$ in (\ref{eq:11}) satisfies $\cost(C, S^\star) \le (1+2/e+O(\e))\;\cost(C, F^\star)$.
\end{lemma}

\begin{proof}
We bound the cost associated with this solution as follows. 
\begin{align}
  \cost(C,S^\star) \qquad &
                   \stackrel{\mathclap{(\text{Lem~}\ref{lem:bound-S})}}{=} \quad \cost(F' \cup S^\star) 
  = \cost(C,F') - \impr(S^\star) \notag \\
  & \stackrel{\mathclap{(\ref{eq:11})}}{\leq} \quad \cost(C,F') -
    (1-1/e)\; \impr(F^\star) \notag \\
  &= \quad \cost(C,F') - (1-1/e)\;(\cost(C, F')-\cost(C,F^\star)) \notag \\
  &= \quad (1/e)\;\cost(C,F') + (1-1/e)\;\cost(C,F^\star) \notag \\
  &\stackrel{\mathclap{(\text{Lem~}\ref{lem:bound-OPT})}}{\leq} \quad (3+2\e)(1/e)\;\cost(C,F^\star)+(1-1/e)\;\cost(C,F^\star) \label{eq:mixie}\\
  &= \quad (1+2/e+O(\e)) \;\cost(C,F^\star). \notag
\end{align}
Hence the proof.
\end{proof}

\subsubsection{Putting it all together}

Our algorithm is a Monte Carlo randomized algorithm: both our
subroutines use randomness. The first is the core-set construction in
\S\ref{sec:coresets}, and the second is the submodular maximization
procedure in Step~\ref{line:submodmax} of the algorithm. For each, we
can make the error probability $1/\poly(n)$. Since each iteration of the
algorithm can be implemented in $\poly(n)$ time, the runtime is
dominated by the number of iterations, which is
$(O(\e^{-2} k\log k)^{k} \poly(n))$. Moreover, combining the two steps
of finding the core-set and the submodular maximization,
 the approximation ratio is
$(1+\e)(1+2/e + O(\e)) = 1 + 2/e + O(\e).$ This proves
Theorem~\ref{thm:main-alg} for the \kmed problem.

\input{fpt-k-cover-hardness}



\bibliography{references}

\appendix
\input{appendix-misc}

\section{Extensions to Related Problems}
\label{sec:extensions}

\input{kmeans}

\input{facility-location}

\input{appendix-hardness}

\end{document}

%% file: intro.tex
\section{Introduction}

How well can we approximate the \kmed and \kmeans clustering problems?
This question has been intensively studied over the past two decades,
and many interesting algorithmic techniques have been developed and
refined in an attempt to understand these problems. Let us elaborate for
the \kmed problem; the story for \kmeans is much the same. Recall that
in the \kmed problem, given a metric space $(V,d)$ with $n$
points and clients at some of the points, the goal is to open $k$
\emph{facilities} such that the sum of distances from the clients to
their closest facilities is minimized.

The first constant-factor approximation algorithm for \kmed was given by
Charikar et al.~\cite{CharikarGTS02}. After many interesting
developments (e.g., primal-dual schemes, sophisticated LP rounding
schemes, and pseudo-approximations), today the best approximation
guarantee is 2.611~\cite{byrka14}. The best lower bound, however, is
still the $(1+2/e)$-hardness from 1998, due to Guha and
Khuller~\cite{GK98}. In this paper, we ask: can we do better if we give
ourselves more resources?  The problem can be solved exactly by
brute-force enumeration in time $n^{k+O(1)}$, but what can we do, say, in
FPT time $f(k) n^{O(1)}$?

We cannot hope to solve the problem exactly in FPT time: the reduction
of Guha and Khuller also shows a $W[2]$-hardness for finding the optimal
solution for \kmed/\kmeans exactly.  Naturally, we then ask what we can
achieve by combining the two approaches together, and whether good
approximation algorithms can be given in FPT time.




\medskip\textbf{Our Results.}
 Our main algorithmic result is a positive result in this direction:
\begin{theorem}[Algorithm for \kmed/\kmeans]
  \label{thm:main-alg}
  For every $\e > 0$, there is a $(1+2/e+\e)$-approximation algorithm
  for the \kmed problem, that runs in time FPT time, i.e., in
  $f(k,\e)n^{O(1)}$ time.  For the \kmeans problem, we can achieve a
  $(1+8/e+\e)$-approximation in the same runtime.
\end{theorem}

The approximation guarantees in Theorem~\ref{thm:main-alg} match the
NP-hardness results for the two problems implied by~\cite{GK98}. However,
since we are allowing ourselves FPT time and not just $\poly(n,k)$ time,
can we do even better and go past this NP-hardness barrier? Our second main result shows
that this is not possible, at least under recent complexity-theoretic
conjectures. We prove that the results in Theorem~\ref{thm:main-alg} are
essentially tight, assuming the Gap-Exponential Time
Hypothesis~\cite{Dinur16,MR17,CCKLMNT17}:
\begin{restatable}[Hardness]{theorem}{kmedHardness}
  \label{thm:hard}
  There exists a function $g : \R^+ \to \R^+$ such that assuming the
  Gap-ETH, for any $\eps > 0$, any $(1 + 2/e - \eps)$-approximation
  algorithm for \kmed, and any $(1+8/e - \eps)$-approximation for
  \kmeans, must run in time at least $n^{k^{g(\eps)}}$.
\end{restatable}

The basic component of the above hardness result is an FPT-hardness of a
factor of $(1-1/e)$ for the \mkc problem, again using the Gap-ETH
(Theorem~\ref{thm:mkc-hardness}). Composing that hardness result with
the reduction of Guha and Khuller~\cite{GK98} gives us
Theorem~\ref{thm:hard} above.

\medskip\textbf{Matroid Median.}
Finally, using our algorithmic techniques, we are able to also give an
improved approximation for the matroid-median problem, which is a
generalization of the \kmed problem. 
\begin{theorem}[Algorithm for \matmed]
  \label{thm:matroid-median}
  There is a $(2+\e)$-approximation algorithm for the \matmed  problem,
  that runs in time FPT time, i.e., in $f(k,\e)n^{O(1)}$ time.
\end{theorem}

Since the \matmed problem is a generalization of the \kmed problem, the
$(1+2/e -\eps)$-hardness from Theorem~\ref{thm:hard} translates
immediately to \matmed.  It remains an open problem to close the gap
between this lower bound and the $(2+\e)$-approximation in
Theorem~\ref{thm:matroid-median}. 
We can also use our ideas to get an $(3-\frac{2}{p+1} + \e)$ for the
$p$-\matmed problem. 

\medskip\textbf{\fl.} \fl is a problem closely related to \kmed, 
where each facility has an opening cost and the goal is to open facilities to minimize the sum of distances from clients to their closest facilities plus the sum of the total opening costs. For this problem, the best known hardness ratio is $\flnumber \approx 1.463$~\cite{GK98},
which is defined to be $\max_{x \geq 0} \big( 1 + \frac{x}{1 + x} \ln \frac 2x \big)$. 
On the other hand, the best algorithm achieves an $1.488$-approximation~\cite{Li13}.
When the parameter $k$ denotes the number of facilities open in the optimal solution,  we prove that our techniques also give an FPT algorithm for \fl whose approximation ratio matches the hardness ratio of~\cite{GK98}. 

\begin{theorem}[Algorithm for \fl]
  \label{thm:fl}
  There is a $(\flnumber+\e)$-approximation algorithm for the \fl  problem,
  that runs in time FPT time, i.e., in $f(k,\e)n^{O(1)}$ time.
\end{theorem}

\elnote{Hardness carried over? How to say it?}

\textbf{Roadmap:} In Section~\ref{sec:approx}, we describe the
approximation algorithms for these problems. We assume throughout that
the aspect ratio is polynomially bounded. (We show in
Section~\ref{sec:aspect-ratio} that this assumption is without loss of
generality, in the case we consider where the clients have unit weights.)
In Section~\ref{sec:hardness}, we then give the hardness results for FPT
\mkc, \kmed, and \kmeans.

\subsection{Our Techniques}

The algorithm is inspired by the hardness result from~\cite{GK98}: it
relies on the result of Feige~\cite{Feige98} that \mkc is hard to
approximate better than $(1-1/e)$. Hence, if we build a ``factor graph''
with sets on one side and elements on another, with edges indicating
inclusion, picking $k$ sets covers $(1-1/e)$ elements at distance $1$,
and the remaining at distance at least $3$---hence $1+2/e$. Now what if
we have a  general instance, with different distances? We show how to do
limited enumeration (in FPT) time to restrict our choices to picking one
facility each from $k$ disjoint sets. Moreover, via a surprisingly clean
idea we can model the objective as submodular maximization (subject to a
partition matroid constraint). And this problem can be approximated
well: the factor again is $(1-1/e)$, hence giving the same factor upto
additive $\e$ terms!

The matching hardness result is via showing an FPT hardness for \mkc
assuming the Gap-ETH. Firstly, we show that assuming the Gap-ETH, there
is no FPT approximation algorithm for \lc problem parameterized by the
number of vertices $k$ on one side of the bipartition. (Trying all
labelings on one side takes time $O(n^{k+O(1)})$, and doing much better
is hard.)  To do this, we construct a {\em variable-clause game} from a
$3$-SAT instance, merge {\em clause vertices} into $\ell$
super-vertices, and then use $r$ rounds of parallel repetition. (The
number of clause vertices becomes $k := \ell^r$.) Then we compose this
with the classical reduction from \lc to \mkc~\cite{Feige98}. Due to
some technical details (e.g., our \lc instance is not guaranteed to be
regular) and for the sake of completeness, we provide a formal proof in
Lemma~\ref{lem:mkc}. While our techniques are similar to recent FPT
hardnesses for the related \kdominatingset problem~\cite{CCKLMNT17,
  CLM18}, some technical details (e.g., the {\em projection property} of
\lc instances) prevent us from directly using prior results to get $(1-1/e+\eps)$-hardness for \mkc.

\subsection{Related Work}
\label{sec:related-work}

We briefly survey the state-of-the-art for \kmed and \kmeans; please see
references below for more historical context. For general metric spaces,
the best approximation ratio for \kmed is 2.611~\cite{byrka14} by Byrka
et al., building on work of Li and Svensson~\cite{LS16}. 
Kanungo et
al.~\cite{KanungoMNPSW04} gave a $(9+\e)$-approximation algorithm for
\kmeans in general metric spaces, which was later improved to 6.357 by
Ahmadian et al.~\cite{AhmadianNSW17}. The first constant factor
approximation algorithm for \matmed was given by Krishnaswamy et
al.~\cite{Krishnaswamy0NS15}, which was improved by Swamy~\cite{Swamy16} to~8.

For Euclidean spaces, the problems are better approximable, at least
when either $k$ or the dimension $d$ are fixed; we restrict this
discussion to parameterizing by $k$. Specifically, PTASs for both \kmed
and \kmeans with running time $f(k, \e) \poly(n,d)$ were given by Kumar
et al.~\cite{KumarSS10}. The running times were improved by
Chen~\cite{Chen2006} to $O(nkd+d^2 n^\sigma 2^{(k/\e)^O(1)})$ for any
$\sigma > 0$ for \kmed, and by Feldman et al.~\cite{feldman2013} to
$O(nkd+d \poly(k/\e) + 2^{{\tilde O} (k/\e)})$ for \kmeans. Both these
latter results were based on the notion of coresets. The \kmeans problem
is APX hard even in Euclidean space, if both $k$ and $d$ are allowed to
be arbitrary~\cite{AwasthiCKS15, LeeSW17}.

A result of direct interest to this work is that of Czumaj and
Sohler~\cite{CzumajS10}, for the \emph{min-sum clustering problem}. They
give a $(4+\e)$-approximation on general metrics in FPT time.  They
construct a small (strong) core-set for the related \textsc{Balanced
  $k$-Median} problem, and enumerate over all choices of centers inside
this core-set. We show in \S\ref{sec:bipartite} that their approach
extends to give a $2$-approximation for the \emph{non-bipartite case} of
\kmed --- in this special case of \kmed a facility may be opened at any
client location, and hence $C \sse \F$. Theorem~\ref{thm:main-alg} above
shows how to get a better guarantee for a more general case. (As an
aside, the hardness for this special non-bipartite case is only
$(1+1/e)$; closing this gap is another interesting open question.)

Hardness-of-approximation results for parameterized problems have been
actively studied recently. Lin~\cite{Lin15} proved $W[1]$-hardness of
approximation for \kbiclique. Chen and Lin~\cite{ChenL16} proved
$W[1]$-hardness of approximation for \kdominatingset in any constant
factor, which was later improved to any function $f(k)$
in~\cite{CCKLMNT17, CLM18}. Chalermsook et al.~\cite{CCKLMNT17} also
proved that there is no FPT $o(k)$-approximation algorithm for \kclique
assuming the Gap-ETH.


%% file: prelims.tex
\subsection{Preliminaries}
\label{sec:preliminaries}

An instance $\I$ of the \kmed problem is defined
by a tuple $((V,d), C, \F, k)$, where $(V,d)$ is a metric space over a
set of points $V$ with $d(i,j)$ denoting the distance between two points
$i,j$ in $V$. Further, $C$ and $\F$ are subsets of $V$ and are referred
as ``clients'' and ``facility locations'', and $k$ is a positive
parameter.  The goal is to find a subset $F$ of $k$ facilities in $\F$
to minimize
$$\cost(C, F) := \sum_{j \in C} d(j,F). $$ In the
\emph{weighted} version of \kmed, every client $j \in C$ has an
associated weight $w_j$, and the goal is to find a subset $F$ of $\F$ of
size $k$ such that $\cost(C, F) := \sum_{j \in C} w_j d(j,F) $ is
minimized.

The \kmeans problem is defined similarly except that the
objective function gets modified to
$\cost(C, F) := \sum_{j \in C} d(j,F)^2$ (and analogously for the
weighted version). The names of the two problems come from the fact that
if the metric space is the real line and $k=1$, the optimal solution is
the median and the mean respectively.  In the \matmed problem, we are
given a matroid on the set $\F$, and the set of open facilities must be
an independent set in the matroid. Again, the goal is to minimize the
assignment cost of clients to the nearest open facility.

In the \fl problem, an instance is not given $k$, but additionally has 
$\open : \F \to \R^+$ that 
indicates the opening cost of each facility. The goal is to find a subset 
$F \subseteq \F$ (without any restriction on $|F|$) that minimizes
$\open(F) + \cost(C, F)$ where $\open(F) := \sum_{f \in F} \open(f)$. 

Finally, the \emph{aspect ratio} of a metric space $(V,d)$ is
$\Delta := \frac{\max_{x,y \in V} d(x,y)}{\min_{x,y \in V} d(x,y)}$.

%% file: fpt-k-cover-hardness.tex
\section{Gap-ETH Hardness of Max $k$-Coverage}
\label{sec:hardness}

In this section, we show that assuming the Gap Exponential Time
Hypothesis (Gap-ETH)~\cite{Dinur16, MR17}, for any $\eps > 0$, there is
no FPT-approximation algorithm that approximates \mkc better than a
factor $(1 - 1/e + \eps)$.

\begin{restatable}[Hardness for Max-Coverage]{theorem}{kcovHardness}
  \label{thm:mkc-hardness}
  There exists a function $g : \R^+ \to \R^+$ such that assuming the
  Gap-ETH, for any $\eps > 0$, any $(1 - 1/e + \eps)$-approximation
  algorithm for \mkc with $n$ elements and $m$ sets must run in time at
  least $(n + m)^{k^{g(\eps)}}$.
\end{restatable}

Using the reduction of Guha and Khuller~\cite{GK98}, this immediately
implies Theorem~\ref{thm:hard}. 
The rest of the section is devoted to the proof of
Theorem~\ref{thm:mkc-hardness}.  The proof has two main components: the
first part shows under the Gap-ETH, it takes at least $n^{h(k)}$ time to
approximate the \lc problem even when one side of the bipartition has
only $k$ vertices; here $h(\cdot)$ is some increasing function depending
on the quality of approximation. This reduction is inspired by the
recent progress on the hardness of parameterized
problems~\cite{CCKLMNT17, CLM18} and was communicated to us by Pasin
Manurangsi.  The second part is the classical reduction from \lc to \mkc
given by Feige~\cite{Feige98}.

\subsection{Hardness of \lc from Gap-ETH}
\label{sec:lc-from-GETH}

We begin with the standard definition of \lc. 
\begin{definition}[Label Cover]
  An instance of \lc $\call$ consists of a bipartite graph
  $G = (U \cup V, E)$ with possibly parallel edges, two label sets
  $\Sigma_U, \Sigma_V$, and a projection $\pi_e : \Sigma_U \to \Sigma_V$
  for each $e \in E$.  Given a labeling
  $\sigma : (U \cup V) \to (\Sigma_U \cup \Sigma_V)$, an edge
  $e = (u, v) \in E$ is satisfied when $\pi_e(\sigma(u)) =
  \sigma(v)$. The goal of \lc is to find a labeling $\sigma$ that
  maximizes the number of satisfied edges. Let $\OPT(\call)$ be the
  maximum fraction of edges simultaneously satisfied by any labeling.
\end{definition}
Note that we include the {\em projection property} in the definition;
all \lc instances in the paper will have this property.  For a vertex
$u \in U \cup V$, let $d_u$ be the degree of $u$, and let $d_U$ (resp.\
$d_V$) be the maximum degree of $U$ (resp.\ $V$).  We also call an
instance \emph{$U$-regular} (resp.\ \emph{$V$-regular}) if all vertices
in $U$ (resp.\ $V$) have the same degree.  All subsequent \lc instances
will be $U$-regular, though the lack of $V$-regularity will require us
to do a little more work in \S\ref{sec:feige}.

Given a 3-SAT formula $\phi$, let $\OPT(\phi)$ be the maximum fraction
of clauses that can be satisfied by any assignment.  The
Gap-ETH~\cite{Dinur16, MR17} states that there exist some constants
$\delta > 0, s < 1$ for which no algorithm, given a 3-SAT formula $\phi$
on $n$ variables and $m = O(n)$ clauses, can distinguish whether
$\OPT(\phi) = 1$ or $\OPT(\phi) < s$ in time $O(2^{\delta n})$.  The
main result of this subsection is the following lemma.

\begin{lemma}
  For every $\ell, r \in \mathbb{N}$, there is a reduction that, given
  3-SAT formula $\phi$ with $n$ variables and $m$ clauses, outputs a
  $U$-regular \lc instance $\call$ such that
  \begin{itemize}
  \item (Completeness) $\OPT(\phi) = 1 \implies \OPT(\call) = 1$, and 
  \item (Soundness) $\OPT(\phi) < s \implies \OPT(\call) < s^{\Omega(r)}$,
  \end{itemize}
  where
  $|U| = \ell^r, |V| = n^r, |\Sigma_{U}| = 2^{O(mr / \ell)},
  |\Sigma_{V}| = 2^{O(r)}, d_V \leq m^r$. The running time of this
  reduction is $m^{O(r)} \cdot |\Sigma_U|$.

  In particular, assuming
  Gap-ETH, for any $\eta > 0$, if we let $r = \Theta(\log (1 / \eta))$ 
  so that 
\[
|\Sigma_U|^{|U|^{O(1 / \log (1 / \eta))}} = 
|\Sigma_U|^{|U|^{1 / 2r}} =
|\Sigma_U|^{\ell^{1/2}} = 
2^{O(mr / \sqrt{\ell})}, 
\]
no
  algorithm can take a \lc instance $\call$ and can decide whether
  $\OPT(\call) = 1$ or $\OPT(\call) < \eta$ in time
  $|\Sigma_U|^{|U|^{O(1 / \log (1 / \eta))}}$.
\label{lem:lc}
\end{lemma}

Note that a brute-force algorithm that tries every assignment to $U$ and
chooses the best assignment for $V$ for it runs in $O(|\Sigma_U|^{|U|})$
times a polynomial.  Lemma~\ref{lem:lc} shows that assuming the Gap-ETH,
even approximately solving \lc requires significant time.

Lemma~\ref{lem:lc} is proved by a series of well-known transformations
between \lc instances.  We start with the following basic hardness
result for \lc assuming the Gap-ETH, which follows from essentially
restating Gap-ETH as a {\em clause-variable} game:
\begin{theorem}[Theorem 4.1 of~\cite{CCKLMNT17}]
  There is a reduction that, given 3-SAT formula $\phi$ with $n$
  variables and $m$ clauses, outputs a $U$-regular \lc instance $\call$
  such that
  \begin{itemize}
  \item (Completeness) $\OPT(\phi) = 1 \implies \OPT(\call) = 1$, and 
  \item (Soundness) $\OPT(\phi) < s' \implies \OPT(\call) < s = 1 - (1 - s') / 3$,
  \end{itemize}
  where
  $|U| = m, |V| = n, |\Sigma_{U}| = 7, |\Sigma_{V}| = 2, d_V \leq m$,
  and $G$ is $U$-regular with $d_U = 3$.  In particular, assuming the
  Gap-ETH, there exist constants $\delta > 0$, $s < 1$ such that no
  algorithm can take a \lc instance $\call$ and can decide whether
  $\OPT(\call) = 1$ or $\OPT(\call) < s$ in $O(2^{\delta |U|})$ time.
\label{thm:lc_0}
\end{theorem}

Let $\ell$ be a parameter that will be related to $k$ in \mkc later.  We
can ensure $\ell$ divides $|U|$ by taking an arbitrary vertex in $U$ and
making $\ell \lceil |U| / \ell \rceil - |U|$ copies of it.  This does
not change any of the properties in Theorem~\ref{thm:lc_0} except to
increase the soundness $s$ by $o_n(1)$; however, the soundness still
remains bounded away from $1$.

Since we want few vertices on the left, we construct a new \lc instance
$\call_1$ by partitioning $U$ into $\ell$ groups and creating
super-vertices for each one. Formally, index the vertices of $u$ as
$U = \{ u_{i, j} \}_{i \in [\ell], j \in [m / \ell]}$, and let the
$i^{th}$ part be $S_{i} := \{ u_{i, j} \}_{j \in [m / \ell]}$. The new
instance
$\call_1 = ((U_1 \cup V_1, E_1), \Sigma_{U_1}, \Sigma_{V_1}, \{
\pi^1_{e} \}_{e \in E_1})$ is constructed as follows.

\begin{itemize}
\item $V_1 = V$ and $\Sigma_{V_1} = \Sigma_{V}$ (the RHS remains unchanged), 
\item $U_1 = \{ S_1, \dots, S_{\ell} \}$. $\Sigma_{U_1} = (\Sigma_U)^{m /
    \ell}$ (the LHS has one super-vertex for each group), and 
\item for each $e = (u, v) \in E$ such that $u = u_{i, j}$, add an edge
  $e' = (S_i, v)$ to $E_1$ with the projection
  $\pi^1_{e'}(\sigma_1, \dots, \sigma_{m / \ell}) := \pi_e(\sigma_j)$
  where the latter $\pi_e$ denotes the projection in $\call$. (Recall we
  allow parallel edges with different projections.)
\end{itemize}

Since the set of possible labelings and the set of edges remain the same
except for syntactic changes, the completeness $c$ and the soundness $s$
do not change. The parameters become
$|U_1| = \ell, |V_1| = |V|, |\Sigma_{U_1}| = 2^{O(m / \ell)},
|\Sigma_{V_1}| = O(1)$. It still maintains $U$-regularity and
$d_{V_1} = d_V \leq m$.

The final transformation is the powerful {\em parallel repetition} step,
which shows that the soundness decreases exponentially as we take the
natural graph power. Fix $r \in \mathbb{N}$. The instance
$\call_2 = ((U_2 \cup V_2, E_2), \Sigma_{U_2}, \Sigma_{V_2}, \{
\pi^2_{e} \}_{e \in E_2})$ is constructed as follows.
\begin{itemize}
\item $U_2 = (U_1)^r$ and $\Sigma_{U_2} = (\Sigma_{U_1})^r$. 
\item $V_2 = (V_1)^r$ and $\Sigma_{V_2} = (\Sigma_{V_1})^r$. 
\item $E_2 = (E_1)^r$. For each $e = (e_i)_{i \in [r]} \in E_2$ with
  $e_i = (u_i, v_i) \in E_1$ and
  $(\sigma_1, \dots, \sigma_r) \in \Sigma_{U_1}^r$,
  $\pi^2_e(\sigma_1, \dots, \sigma_r) = (\pi^1_{e_1}(\sigma_1), \dots,
  \pi^1_{e_r}(\sigma_r))$.
\end{itemize}
The parameters become
$|U_2| = \ell^r, |V_2| = |V|^r = n^r, |\Sigma_{U_2}| = 2^{O(mr / \ell)},
|\Sigma_{V_2}| = 2^{O(r)}, d_{V_2} \leq d_{V_1}^r \leq m^r$, and
$\call_2$ maintains $U$-regularity. The completeness $c$ still remains
$1$, and by the parallel repetition theorem~\cite{Raz98}, the soundness
drops $s = 2^{- \Theta(r)}$, where the constant hiding in the $\Theta$
depends on the original soundness. This proves Lemma~\ref{lem:lc}.

\subsection{Hardness of \mkc from \lc}
\label{sec:feige}

Given the ``nice'' \lc instance from Lemma~\ref{lem:lc} we now show how
to reduce this to \mkc. This reduction is standard and closely follows
the classical one given by Feige~\cite{Feige98}, modulo some minor
issues arising from it not being $V$-regular.

Recall that an instance of $\cali$ of \mkc consists of an underlying
universe $\calu$, a family $\cals$ of subsets, and an integer $k$. The
goal is to find a subfamily $\cals' \subseteq \cals$ with $|\cals'| = k$
that covers the largest number of elements. For notational simplicity,
we prove the hardness of the {\em weighted version} of \mkc where each
element $e \in \calu$ has weight $w(e)$ and we want to maximize the
total weight of the covered elements.  Note that weighted instances can
be easily converted to unweighted instances by duplicating elements
according to their weights. 
In our reduction, the ratio between the maximum and the minimum weight
will be bounded by the number of elements. 
The proof appears in Section~\ref{appendix-hardness}.

\begin{lemma}[Reduction \#2]
  There exist functions $a : \R^+ \to \mathbb{N}$ and
  $f : \R^+ \to \R^+$ such that for any $\eps > 0$, there exists a
  polynomial-time reduction that takes a \lc instance
  $\call = ((U \cup V, E), \Sigma_U, \Sigma_V, \{ \pi_e \}_{e \in E})$
  that is $U$-regular and has the maximum $V$-degree $d_V$, and produces
  a \mkc instance $\cali = (\calu, \cals, k)$ such that
  \begin{itemize}
  \item (Completeness)
    $\OPT(\call) = 1 \implies \OPT(\cali) = w(\calu)$.
  \item (Soundness)
    $\OPT(\call) < f(\eps) \implies \OPT(\cali) \leq (1 - 1/e + \eps)
    \cdot w(\calu)$.
  \end{itemize}
  The reduction satisfies
  $|\calu| \leq |V| \cdot |d_V|^{a(\eps)} \cdot a(\eps)^{\Sigma_V}$,
  $|\cals| = a(\eps) \cdot |U| \cdot |\Sigma_U|$, and $k = a |U|$.
  \label{lem:mkc}
\end{lemma}

We can now finish the proof of Theorem~\ref{thm:mkc-hardness} based on Lemma~\ref{lem:lc} and Lemma~\ref{lem:mkc}.
\begin{proof}[Proof of Theorem~\ref{thm:mkc-hardness}]
  Fix $\eps > 0$ that determines $a(\eps)$ and $f(\eps)$ in
  Lemma~\ref{lem:mkc}. Let $r \in \mathbb{N}$ in Lemma~\ref{lem:lc} so
  that the soundness $2^{-\Omega(r)} \leq f(\eps)$.

  With $\ell$ still being a free parameter, Lemma~\ref{lem:lc} shows a
  reduction from an initial 3-SAT instance $\phi$ with $n$ variables and
  $m = O(n)$ clauses to a \lc instance with
  $|U| = \ell^r, |V| = n^r, |\Sigma_{U}| = 2^{O(mr / \ell)},
  |\Sigma_{V}| = 2^{O(r)}$, and $d_V \leq m^r$.  Lemma~\ref{lem:mkc}
  with this \lc instance produces a \mkc instance with
  \begin{align*}
    |\calu| &\leq |V| \cdot |d_V|^{a} \cdot a^{\Sigma_V} = n^{O(ar)} \cdot a^{2^{O(r)}}  \\
    |\cals| &= a \cdot |U| \cdot |\Sigma_U| = a \ell^r \cdot 2^{O(nr / \ell)} \\
    k &= a |U| = a \ell^r. 
  \end{align*}

  An $(1 - 1/e + \eps)$-approximation algorithm for \mkc that runs in
  time $|\cals|^{k^{1/2r}}$ will distinguish whether $\OPT(\phi) = 1$ or
  $\OPT(\phi) < s'$ for some $s'$ in time
  \[
    2^{O((nr / \ell) \cdot k^{1/2r})} = 
    2^{O((nr / \ell) \cdot \sqrt{\ell})} =
    2^{O(nr / \sqrt{\ell})},
  \]
  which will contradict the Gap-ETH for large enough $\ell$. Observe
  that $|\cals| \gg |\calu|$; if we set $g(\eps) := 1/2r$, we get the
  same implication from an algorithm that runs in time
  $|\calu|^{k^{g(\e)}}$, which proves the theorem.
\end{proof}




%% file: appendix-misc.tex
\section{Omitted Proofs}
\label{sec:omitted-proofs}

\Metric*

\begin{proof}
Since the distances between points in the original metric space do not change, we only need to check triangle inequalities involving fictitious centers.
We prove by induction on $i = 0, 1, \dots, k$ that the distances on $V \cup \{ f'_1, \dots, f'_i \}$ form a metric. The base case $i = 0$ holds since $V$ is metric. 

For general $i$, let $f'_i \in F'$ be a fictitious center and $u, v \in V \cup \{ f'_1, \dots, f'_{i - 1} \}$ be arbitrary points. We consider the following two cases. 
\begin{itemize}
\item For $d(f'_i, v) \leq d(f'_i, u) + d(u, v)$, 
\[
d(f'_i, v) = \min_{f \in F_i} (2R_i + d(f, v)) \leq \min_{f \in F_i} (2R_i + d(f, u) + d(u, v)) = d(f'_i, u) + d(u, v),
\]
where the inequality follows from the triangle inequality between $f, u, v$. 

\item For $d(u, v) \leq d(u, f'_i) + d(f'_i, v)$, first note that 
\[
d(u, f'_i) + d(f'_i, v) 
= \min_{f \in F_i} (2R_i + d(u, f)) + \min_{f \in F_i} (2R_i + d(f, v)).
\]
Let $g$ and $h$ be the facilities achieving the minimum in the first and the second minimization respectively. 
Since $g$ and $h$ are both in $F_i$, $d(g, h) \leq 2R_i$. Therefore, 
\[
d(u, f'_i) + d(f'_i, v)  = 4R_i + d(u, g) + d(h, v) \geq d(g, h) + d(u, g) + d(h, v) \geq d(u, v). 
\]
\end{itemize}
Therefore, all triangle inequalities are satisfied and the new distance on $F' \cup V$ is a metric. 
\end{proof}

\Submod*

\begin{proof}
We have $\impr(\emptyset)=0$ by definition. To show that $\impr(S)$ is monotone, consider two subsets $S\s T\s F$:
\begin{align*}
\cost(F'\cup T)&=\sum_{j\in C'}w_j \; d(j,F'\cup T) \le \sum_{j\in C'}w_j \; d(j,F'\cup S)=\cost(F'\cup S)
\\ &\stackrel{(\text{Step~}\ref{line:impr})}{\implies} \impr(S)\le\impr(T),
\end{align*}
as desired. Finally, to prove that $\impr(S)$ is submodular, consider subsets $S\s T\s F$ and center $f\in F$. For each client $j\in C'$, using the identity $x-\min(x,y)=\max(0,x-y)$ for all real numbers $x$ and $j$, we get
\begin{align*}
d(j,F'\cup S)-d(j,F'\cup (S\cup\{f\}))   &= d(j,F'\cup S)-\min(d(j,F'\cup S),\; d(j,\{f\}))
\\&=  \max(0,\;d(j,F'\cup S)-d(j,\{f\}))
\\&\ge \max(0,\;d(j,F'\cup T)-d(j,\{f\})) 
\\&=  d(j,F'\cup T)-\min(d(j,F'\cup T),\; d(j,\{f\}))
\\&= d(j,F'\cup T)-d(j,F'\cup (T\cup\{f\})),
\end{align*}
Therefore,
\begin{align*}
\impr(S\cup\{f\})-\impr(S) &=\cost(F'\cup S)-\cost(F'\cup (S\cup\{f\})) \\&= \sum_{j\in C'}w_j(d(j,F'\cup S)-d(j,F'\cup(S\cup\{f\})))
\\&\ge \sum_{j\in C'}w_j(d(j,F'\cup T)-d(j,F'\cup(T\cup\{f\})))
\\&=\cost(F'\cup T)-\cost(F'\cup (T\cup\{f\}))
\\&=\impr(T\cup\{f\})-\impr(T),
\end{align*}
proving the desired submodularity.
\end{proof}

\section{Miscellaneous Proofs}

\subsection{Polynomial Aspect Ratio}
\label{sec:aspect-ratio}

Recall that the \emph{aspect ratio} of a metric space $(V,d)$ is
$\Delta := \frac{\max_{x,y \in V} d(x,y)}{\max_{x,y \in V} d(x,y)}$. For
the unweighted version of the problems we consider, we can assume that
$\Delta$ is polynomially-bounded, due to the following standard result.
\begin{proposition}[folklore]
  Given an $\alpha$-approximation algorithm $A$ for (unweighted) \kmed
  on instances with polynomially-bounded aspect ratio that runs in time
  $T$, we can obtain an $(\alpha+ o(1))$-approximation algorithm $B$ for
  (unweighted) \kmed on all instances running in time $T + \poly(n)$.
\end{proposition}

\begin{proof}
  Given an instance $I$ with large aspect ratio, we first compute a
  estimate $M$ for the optimal \kmed cost on $I$---say
  $M/(2n) \leq OPT(I) \leq M$, by using an approximation algorithm for
  the \textsc{$k$-Center} problem that runs in $\poly(n)$ time for
  general instances. View the metric space $(V,d)$ as a complete
  edge-weighted graph. For \emph{long} edges of length more than
  $2\alpha M$, reduce their length to $2\alpha M$, and for \emph{short}
  edges of length less than $M/n^3$, increase their lengths to
  $M/n^3$. Computing all-pairs shortest paths gives a new metric space
  $(V,d')$, and let $I'$ be the corresponding \kmed instance. Use
  algorithm $A$ on this instance $I'$ to get an $\alpha$-approximate
  solution $F \sse V$.

  We claim $F$ is also an $(\alpha+o(1))$-approximate solution for the
  original instance $I$. Firstly, if $F^\star$ is an optimal solution to
  $I$, then its cost in $I'$ is greater by at most $n(M/n^3)$. Indeed,
  since $OPT(I) \leq M$, no client would use the long distances which we
  shortened; the increase in the short distances gives the $n(M/n^3)$
  term. Hence $OPT(I') \leq OPT(I) + n(M/n^3) = OPT(1+o(1))$.  Again
  since $F$ is an $\alpha$-approximation for $I'$, and the long edges
  had reduced length $2\alpha M$, none of the clients in $I'$ will
  connect to it using the shortened long distances. Hence
  \[ \cost_I(C,F) \leq \cost_{I'}(C,F) \leq \alpha\, OPT(I') \leq
    \alpha(1+o(1))\, OPT(I). \]
  This completes the proof.
\end{proof}

\subsection{Bipartite vs. Non-Bipartite Instances}
\label{sec:bipartite}

The \kmed/\kmeans problems we defined have two different sets: clients
$C$ and potential facilities $F$. If $C$ and $\F$ are allowed to be
different subsets of $V$, we call it the \emph{bipartite} version of the
problem. If $C \sse \F$, i.e., we can open facilities at any of the
client locations (and potentially at other locations too), it is the
\emph{non-bipartite case}. We observe that only a $(1+1/e)$-factor
hardness is known for the non-bipartite case, whereas our algorithm
still gives a factor-$(1+2/e)$ approximation for this case.

In fact, for the non-bipartite case, a simple $2(1+\e)^2$-approximation
can be obtained directly using core-sets, using a variant of the arguments of Czumaj and
Sohler~\cite{CzumajS10} as follows. Given an non-bipartite instance $I$, the
algorithm does the following.
\begin{enumerate}
\item Find a core-set $(C',w')$ with $C' \sse C$ and
  $|C'| = O(\poly(\e^{-1} k \log n))$.
\item Enumerate over all subsets $F \sse C'$ being $k$-subsets of $C'$,
  and output the set $F_{alg}$ with smallest cost
  $\sum_{j \in C'} w_j d(j,F)$.
\end{enumerate}
The runtime of this algorithm is easily seen to be in FPT, so we now
show the approximation guarantee. Let $F^\star \sse \F$ be the optimal
solution for instance $I$ with cost $\cost(C,F^\star) = OPT(I)$. By the
strong core-set property,
$\sum_{j \in C'} w_j d(j,F^\star) \leq (1+\e) OPT(I)$. Now, for each
facility $f^\star \in F^\star$ let
$\eta(f^\star) := \{ j \in C' \mid d(j,f^\star) = \min_{f \in F^\star} d(j,f) \}$ be
the closest client among those served by $f^\star$. Observe that
$F' := \{ \eta(f^\star) \mid f^\star \in F^\star\}$ satisfies
$F' \sse C' \sse C \sse \F$, has size $|F'| \leq k$, and ensures that
\begin{gather}
  \sum_{j \in C'} w_j d(j,F') \leq 2\, \sum_{j \in C'} w_j d(j,F^\star) \leq
  2(1+\e) OPT(I). \label{eq:app1}
\end{gather}
The factor of $2$ comes from the fact that
$d(j,\eta(f^\star)) \leq d(j,f^\star) + d(f^\star,\eta(f^\star)) \leq 2d(j,f^\star)$. Now,
since we enumerate over all subsets of $C'$, the cost of the set
$F_{alg}$ is no greater than the LHS of~(\ref{eq:app1}). Again using
the core-set property,
\[ \sum_{j \in C} d(j,F_{alg}) \leq (1+\e)\sum_{j \in C'} w_j d(j,F_{alg})
  \leq (1+\e) \, \sum_{j \in C'} w_j d(j,F') \stackrel{(\ref{eq:app1})}{\leq} 2(1+\e)^2 OPT(I). \]
This completes the proof of the $2(1+\e)^2$ approximation.

Observe that this algorithm crucially uses that $C \sse \F$, so we can
open a facility at the closest client location $\eta(f^\star)$. Hence this
idea does not extend to the bipartite case where $\eta(f^\star)$ may not
belong to $\F$.


%% file: kmeans.tex

\subsection{The Algorithm for \kmeans}
The extension to \kmeans is immediate. The first change is in the
definition of cost: the \kmeans cost is
$\cost(C,F) := \sum_j w_j\,  d(j, F)^2$. However, the induced function is
still monotone submodular. Now, by the calculations identical to
Claim~\ref{lem:bound-OPT}, $d(j,F') \leq (3+2\e)d(j,F^\star)$ for client
$j$; hence
\[ \cost(C,F') \leq (3+2\e)^2 \cost(C,F^\star) = (9+O(\e))\, \cost(C,F^\star). \]
Plugging this into~(\ref{eq:mixie}), we immediately get
\[ \cost(C,S^\star) \leq \cost(C,F^\star) \cdot \Big( (9+O(\e))(1/e) +
  (1-1/e) \Big) \leq (1+\nicefrac{8}{e} + O(\e))\,
  \cost(C,F^\star). \] The runtime is $O(\e^{-4} k \log k)^k \poly(n)$,
which is the same barring  a worse dependence on $\e$ because of the larger
core-set. This proves the result for \kmeans.

\subsection{The Algorithm for \matmed} We follow the algorithm for \kmed,
but now we place two matroid constraints: in addition to the partition
matroid constraint we add in the matroid constraint coming from the
\matmed problem itself. Maximizing a monotone submodular function
subject to two matroid constraints has a
$\frac{1}{2+\e}$-approximation~\cite{LSV10}.  Hence, instead
of~(\ref{eq:mixie}), we get
\[ \cost(C,S^\star) \leq
  (3+2\e)\cdot(\nicefrac{1}{(2+\e)})\cdot\cost(C,F^\star) + (1 -
  \nicefrac1{(2+\e)})\cdot\cost(C,F^\star) \leq (2 + O(\e))\,
  \cost(C,F^\star). \] If the rank of the matroid is $k$, then any valid
base of the matroid is also a $k$-subset; hence a core-set for \kmed is
also a core-set for \matmed. This means the rest of the argument remains
unchanged. 


%% file: facility-location.tex
\subsection{\fl}
\label{sec:fl}

In this subsection, we prove Theorem~\ref{thm:fl} for \fl. Given an
instance $((V, d), C, \F, \open)$ for $\fl$, let $k$ be the number of
facilities opened in the optimal solution. Our parameter will be this
value $k$. Let $\cost^\star$ and $\open^\star$ be the total connection
and opening cost of the optimal solution respectively.  For sake of
simplicity, we assume that $\open(f)$ is the same for every $f$, but our
idea can be easily generalized when facilities have nonuniform opening
costs (by guessing the opening costs of the optimal
facilities to within a $(1+\e)$-factor). \elnote{Explain more?}
\agnote{OK?} This implies that $\open^\star = k\, \open(f)$.

The general structure of the algorithm resembles the algorithm for
\kmed. We first construct a core-set that preserves the connection cost
of every $F \subseteq \F$ with $|F| \leq 2k$, so that we can assume
$|C| = O(\eps^2 k \log n)$.  The algorithm guesses (a)~the {\em leaders}
$\{ \ell_1, \dots, \ell_k \}$, (b)~the {\em distances $R_1, \dots, R_k$
  from the leaders to their facilities} as in Algorithm~\ref{alg:main},
and then (c)~for each $\ell_i$, compute the set of its possible
facilities $F_i$. These sets $F_i$ give us a partition matroid on the
potential facility locations.

\begin{lemma}
  \label{lem:bicrit}
  Consider a monotone submodular function $f$, subject to a partition
  matroid constraint (with rank $k$). There exists a polynomial-time
  algorithm that, given $\gamma \geq 1$, returns a set $S$ with full
  rank and size $|S| \leq \gamma k$, such that for any $F^\star$ with
  $|F^\star| \leq k$, we have
  \[ f(S) \geq \left( 1 - e^{-\gamma} \right) \cdot f(F^\star) . \]
\end{lemma}

\begin{proof}
We first use the algorithm from~\cite{CCPV11} to
  find a set $S_1$ of size $k$ such that $S_1$ is a base of the matroid,
  and
  \begin{gather}
    f(S_1) \geq \left( 1 - 1/e \right) \cdot f(F^\star) . \label{eq:1}
  \end{gather}
  Let $f_{S_1}$ be the residual function defined as $f_{S_1}(S) := f(S \cup S_1) - f(S_1)$. 
  Since $f$ is monotone, we get 
  \begin{gather}
    f_{S_1}(F^\star) \geq f(F^\star) - f(S_1) \label{eq:2}
  \end{gather}
  Now we choose a set $S_2$ by picking $(\gamma-1)k$ more elements that
  greedily maximize the residual function. The analysis of the greedy
  algorithm implies that
  \begin{gather}
    f_{S_1}(S_2) \geq (1 - e^{-(\gamma-1)}) f_{S_1}(F^\star) \stackrel{(\ref{eq:2})}{\geq} 
(1 - e^{-(\gamma-1)}) (f(F^\star) - f(S_1)), 
  \end{gather}
  so that the total cost is at least 
  \begin{gather}
f(S_1) + (1 - e^{-(\gamma-1)}) (f(F^\star) - f(S_1))
= (1 - e^{-(\gamma-1)}) f(F^\star) + e^{-(\gamma-1)} f(S_1)
\stackrel{(\ref{eq:1})}{\geq} 
(1 - e^{-\gamma)} f(F^\star), 
  \end{gather}
which completes the proof.
\end{proof}

We use the algorithm from Lemma~\ref{lem:bicrit} to pick a set $S^\star$
of size $\gamma k$, instead of size $k$ as in
Algorithm~\ref{alg:main}. The opening cost of this solution is
$\gamma \open^\star$, since each facility costs
$\open^\star/k$. Moreover, arguing as in Lemma~\ref{lem:main-lemma} (but
using Lemma~\ref{lem:bicrit} instead of the $(1-1/e)$-approximation
guaranteed by algorithm from~\cite{CCPV11}), the connection cost is
$(1 + 2/e^\gamma)\; \cost^\star$. Several cases arise:
\begin{itemize}
\item If $\open^\star \geq 2e^{-1} \cost^\star$: Trying $\gamma = 1$ gives an approximation ratio at most $\frac{1 + 4/e}{1 + 2/e} \leq 1.424$. 
\item If $\eps^2 \cost^\star \leq \open^\star < 2e^{-1} \cost^\star$: Trying $\gamma = \ln (2 / (\open^\star / \cost^\star))$ gives an approximation ratio $\frac{1 + 2\gamma/e^{\gamma} + 2/e^{\gamma}}{1 + 2/e^{\gamma}}$. 
Recalling $\flnumber := \max_{x \geq 0} \big( 1 + \frac{x}{1 + x} \ln \frac 2x \big)$, by setting $x = 2/e^{\gamma}$, we can see that it is upper bounded by exactly $\flnumber \approx 1.463$. 
\item If $\eps^2 \cost^\star > \open^\star$: Trying $\gamma = 1/\eps$ gives the total cost $(1/\eps)\open^\star + (1+2/e^{1/\eps})\cost^\star \leq (1 + 3\eps) \cost^\star$. 
\end{itemize}
Trying every value of $\gamma \in [1, 1/\eps]$ that makes $\gamma k$ an integer
will achieve an approximation ratio of $\flnumber \approx 1.463$.


%% file: appendix-hardness.tex
\section{Reduction from \lc to \mkc}
\label{appendix-hardness}
In this section, we give a reduction from \lc to \mkc, 
proving Lemma~\ref{lem:mkc}. 

\begin{lemma}[Restatement of Lemma~\ref{lem:mkc}]
  There exist functions $a : \R^+ \to \mathbb{N}$ and
  $f : \R^+ \to \R^+$ such that for any $\eps > 0$, there exists a
  polynomial-time reduction that takes a \lc instance
  $\call = ((U \cup V, E), \Sigma_U, \Sigma_V, \{ \pi_e \}_{e \in E})$
  that is $U$-regular and has the maximum $V$-degree $d_V$, and produces
  a \mkc instance $\cali = (\calu, \cals, k)$ such that
  \begin{itemize}
  \item (Completeness)
    $\OPT(\call) = 1 \implies \OPT(\cali) = w(\calu)$.
  \item (Soundness)
    $\OPT(\call) < f(\eps) \implies \OPT(\cali) \leq (1 - 1/e + \eps)
    \cdot w(\calu)$.
  \end{itemize}
  The reduction satisfies
  $|\calu| \leq |V| \cdot |d_V|^{a(\eps)} \cdot a(\eps)^{\Sigma_V}$,
  $|\cals| = a(\eps) \cdot |U| \cdot |\Sigma_U|$, and $k = a |U|$.
\end{lemma}

\begin{proof}
  The high-level idea of the proof is the following: we choose some
  value $a = a(\e)$. Now the set of elements consists of many disjoint
  hypergrids with sides of size $a$ and with $|\Sigma_V|$ dimensions.
  Indeed, there is a copy of the hypergrid $[a]^{|\Sigma_V|}$ associated
  with each $(v, i)$ pair for $v \in V, i \in [d_v]^a$---one for each
  right vertex and an $a$-subset of its left neighbors.

  Now the sets: they are associated with each $u \in U$ and $j \in [a]$
  and each potential label $\ell \in \Sigma_U$.  The sets associated
  with a pair $(j, u)$ have a nonempty intersection with the hypergrid
  of $(v, i)$ if and only if $u$ is the $i_j$'th neighbor of
  $v$. Indeed, the set for $(j,u,\ell)$ contains the entire $j^{th}$
  ``slice'' of each of these hypergrids, along the $\ell^{th}$
  dimension.  The idea is very clean: if there is a ``good'' labeling
  for $\call$, then all these slices will be chosen in a coordinated way
  along the same dimension, and we will cover all the hypergrids
  completely. If there are no good labelings for $\call$, then these
  slices will be chosen in an uncoordinated way along different
  dimensions, and then we will end up covering only a constant factor of
  the hypergrids. (As intuition, if $a = 2$ and we did not manage to
  pick two slices of the hypercube along the same dimension, we cover
  only $\nicefrac34$ of the cube: the hypergrids allow us to get  $1-1/e$.)

  For those familiar with the exposition from~\cite{Feige98}, we are
  considering an $a$-prover system where the verifier first randomly
  chooses a {\em variable question} $v \in V$ and each of $a$ provers
  gets a {\em clause question} independently sampled from $v$'s
  neighbors.


\noindent{\bf Formal construction.}
  For each $v \in V$, fix an arbitrary ordering of its
  incident edges so that the $d_v$ edges incident on $v$ are represented
  as
  $e_{v, 1} = (\alpha_{v, 1}, v), \dots, e_{v, d_v} = (\alpha_{v, d_v},
  v)$.  Let $a = a(\eps)$ be an integer that will be fixed later, and
  consider the hypergrid $[a]^{|\Sigma_V|}$. Let
  $C_{j, \ell} := \{ (x_1, \dots, x_{|\Sigma_V|} : x_\ell = j \}$ be
  the {\em $j^{th}$ slice in the $\ell^{th}$ coordinate.}
We can describe
  our set system as follows.
  \begin{align*}
    \calu &:= \{ (v, i, x) \mid v \in V, i \in [d_v]^a, x \in
            [a]^{|\Sigma_V|} \} && \mbox{ with weight } w(v, i, x) = 1/ ((d_v)^{a - 1} \cdot |E|), \\
    \cals &:= \{ S(j, u, \ell) \mid j \in [a], u \in U, \ell \in
            \Sigma_U \} && \mbox{ where } S(j, u, \ell) = \{ (v, i)
                           \times C_{j, \pi_{(u, v)} (\ell)} \mid
                           \alpha_{v, i_j} = u, i \in [d_v]^a \}, \\
    k &:= a \cdot |U|.
  \end{align*}

\noindent{\bf Completeness.}
  Suppose the labeling
  $\sigma : (U \cup V) \to (\Sigma_U \cup \Sigma_V)$ satisfies every
  edge of $\call$, then the $k = a |U|$ subsets
  \[ \{ S(j, u, \sigma(u)) \mid j \in [a], u \in U \}\] covers every
  element in $U$; indeed, the element $(v, i, x)$ is covered by the set
  $S(x_{\sigma(v)}, \alpha_{v, i_{x_{\sigma(v)}}}, \sigma(\alpha_{v,
    i_{x_{\sigma(v)}}}))$.  This proves the first claim of the theorem
  that we have perfect completeness.

\noindent{\bf Soundness.}
  For sake of a contradiction assume there
  exists a subfamily $\cals' \subseteq \cals$ such that $|\cals'| = k = a|U|$
  and $\cals'$ covers elements of total weight at least
  $(1 -1/e + \eps)$. Recall that each hypergrid is indexed by a
  $(v, i)$. To simplify notation, we identify a pair $(v, i)$ and its
  hypergrid.  We also define $(v, i)$'s weight
  $w(v, i) := 1 / ((d_v)^{a - 1} |E|)$, which is the weight of each of
  the elements in that hypergrid.  The sum of all hypergrids' weights is
  $\sum_{v \in V} (d_v)^a w(v, i) = \sum_v d_v / |E| = 1$, and let
  $\cald$ be the distribution of $(v, i)$'s according to their weights.
  For the rest of this section, an ``average hypergrid'' refers to a
  random $(v, i)$ sampled from $\cald$, possibly conditioned on
  $(v, i) \subseteq X$ for some subset $X$.

  Recall that each set $S \in \cals$ intersects with one hypergrid in
  exactly one slice $C_{j, \ell}$ or is disjoint from it.  Since the
  \lc instance $\call$ is $U$-regular, the sum of the weights of the
  hypergrids that intersect $S(j, u, \ell) \in \cals'$ is
  \[
    \sum_{v \sim u} \quad \sum_{i \in [d_v]^a \, : \, \alpha_{v, i_j} = u} 1 / (|d_v|^{a - 1} |E|) = d_U / |E|,
  \]
  which means that each set intersects with the same weighted number of
  hypergrids.  Let $t_{v, i}$ be the number of sets in $\cals'$ that
  intersect with the hypergrid $(v, i)$.  By double counting,
  \[
    \E_{(v, i) \sim \cald} [t_{v, i}] = 
    \sum_{(v, i)} t_{v, i} / (|d_v|^{a-1} |E|) = |\cals'| \cdot d_U / |E| = a \cdot (d_U |U| / |E|) = a.
  \]
  So each hypergrid intersects with $a$ sets from $\cals'$ in average.

  Call a hypergrid $(v, i)$ {\em big} when $t_{v, i} > 3a / \eps$, and
  call $(v, i)$ {\em good} if it is not big and there exist
  $j < j' \in [a]$ and $\ell_j, \ell_{j'} \in \Sigma_U$ such
  that
  $\pi_{e_{v, i_j}}(\ell_j) = \pi_{e_{v, i_{j'}}}(\ell_{j'})$
  and both $S(j, u_{v, i_j}, \ell_j)$ and
  $S(\ell, u_{v, i_{j'}}, \ell_{j'})$ are in $\cals'$. In other
  words, hypergrid $(v, i)$ is intersected in at least two different
  slices in the same coordinate.  Call the remaining $(v, i)$'s {\em
    pseudorandom}.

  Since the average of $t_{v, i} = a$, the total weight of big
  $(v, i)$'s is at most $\eps / 3$. Hence elements of total weight at
  least $(1 - 1/e + 2\eps / 3) \cdot w(\calu)$ must be covered in the
  good or pseudorandom hypergrids.  The average value of $t_{v, i}$ for
  good and pseudorandom $(v, i)$ hypergrids is still at most $a$.

  We claim that the total weight of good $(v, i)$'s is at least
  $\eps / 3$.  Suppose not. Then elements of total weight at least
  $(1 - 1/e + \eps / 3) \cdot w(\calu)$ are covered in the pseudorandom
  hypergrids.  Note that the average value of $t_{v, i}$ for the
  pseudorandom pairs is at most $(1 + \eps / 3)a$.  For each of those
  hypergrids, since it is not good, the fraction of points covered by
  $a'$ slices is exactly $(1 - (1 - 1/a)^{a'})$, which is monotone and
  concave in $a'$.  Therefore, the fraction of points covered in the
  pseudorandom cubes is at most $(1 - (1 - 1/a)^{(1 + \eps/3)a})$.  Fix
  $a = a(\eps)$ large enough so that this quantity becomes less than
  $(1 - 1/e + \eps/3)$, leading to the desired contradiction.
  Therefore, the total weight of good $(v, i)$'s is at least $\eps / 3$.

  For $j \in [a]$ and $u \in U$, let
  $\Sigma(j, u) := \{ \ell \in \Sigma_U \mid (j, u, \ell) \in \cals' \}$
  be the labels that correspond to $(j,u)$. We now construct a random
  labeling $\sigma$ for $\call$ as follows.
  \begin{itemize}
  \item Randomly sample $j < j' \in [a]$ uniformly from among
    $\binom{a}{2}$ unordered pairs.
  \item For $u \in U$, let $\sigma(u)$ be a random label from
    $\Sigma(j, u)$ chosen uniformly and independently (choose an
    arbitrary label if $\Sigma(j, u) = \emptyset$).
  \item For $v \in V$, uniformly sample $i \in [d_v]^{a}$, and let
    $u := \alpha_{v, i_{j'}}$. Let $\ell$ be a random label from
    $\Sigma(j', u)$ chosen uniformly and independently. Let
    $\sigma(v) = \pi_{u, v}(\ell)$. (Choose an arbitrary label if
    $\Sigma(j', u) = \emptyset$).
  \end{itemize}
  Fix a good pair $(v, i)$. Given that $i$ is sampled in the above
  randomized strategy, with probability at least $1 / \binom{a}{2}$,
  $j < j'$ are sampled such that
  $\pi_{e_{v, i_j}}(\Sigma_{j, \alpha_{v, i_j}}) \cap \pi_{e_{v,
      i_{j'}}}(\Sigma_{j', \alpha_{v, i_{j'}}}) \neq \emptyset$, so
  $\pi_{e_{v, i_j}}(\sigma(\alpha_{v, i_j})) = \sigma(v)$ with
  probability at least
  $(\binom{a}{2} \cdot (3a / \eps)^2)^{-1} \geq \eps^2 / 5 a^3$.

  Fix a vertex $v \in V$ and let $q_v$ be the fraction of
  $i \in [d_v]^{a}$ such that $(v, i)$ is good.  The expected fraction
  of the edges incident on $v$ satisfied by the above randomized
  labeling $\sigma$ is
  \begin{align*}
    & \E_{i, j, j'} \Pr_{u : u \sim v} [\pi_{u, v} ( \sigma(u) ) = \sigma(v)] \\
    =& \, \Pr_{i, j, j'} [\pi_{e_{v, i_j}} ( \sigma(\alpha_{v, i_j}) )= \sigma(v)] \\
    \geq& \, q_v \cdot \eps^2 / 5 a^3,
  \end{align*}
  where the first equality follows from the fact that for fixed $j, j'$,
  over the randomness of $i$, $\alpha_{v, i_j}$ and $\alpha_{v, i_{j'}}$
  are sampled uniformly and independently over the neighbors of $v$, so
  that $u$ in the first line can be replaced by $\alpha_{v, i_j}$ in the
  second line.

  Let $\cald_V$ be the distribution over $v \in V$, which is obtained as
  the marginal distribution of $v$ in $\cald$.  This implies that in
  $\cald_V$, $v$ is sampled with probability $d_v / |E|$, and
  \[
    \E_{v \sim \cald_V} [q_v] = \Pr_{(v, i) \sim \cald} [(v,i) \mbox{ is good}] \geq \eps / 3.
  \]
  Therefore, the total fraction of \lc edges satisfied by the above randomized strategy is at least 
  \[
    \E_{v \sim \cald_V} [q_v \cdot (\eps^2 / 5 a^3)] \geq (\eps / 3) \cdot (\eps^2 / 5 a^3) = (\eps/a)^3 / 15.
  \]
  Let $f(\eps) := (\eps/a)^3 / 15$. This choice establishes that
  $\OPT(\cali) > (1 - 1/e + \eps)\, w(\calu) \implies \OPT(\call) \geq
  f(\eps)$, finishing the proof of the soundness claim.
\end{proof}